%% file: Formats12.tex
\documentclass[runningheads,a4paper]{llncs}

\usepackage{amssymb}
\usepackage{graphicx}
\usepackage{verbatim}
\usepackage{triotex}
\usepackage{url}
\usepackage[shortlabels]{enumitem}

\newcommand{\impl}{\ensuremath{\rightarrow}}
\newcommand{\dimpl}{\ensuremath{\leftrightarrow}}
\newcommand{\satL}{\ensuremath{\TRIOop{\vDash}{L}{}}}
\newcommand{\NowST}{\ensuremath{\TRIOop{NowST}{}{}}}
\newcommand{\true}{\ensuremath{\top}}
\newcommand{\false}{\ensuremath{\bot}}
\newcommand{\defeq}{\ensuremath{\stackrel{\textrm{def}}{=}}}
\newcommand{\XTN}{\ensuremath{\textrm{X-TRIO}^{\textrm{F}}_\naturals}}

\DeclareRobustCommand{\Xop}[2]{\ensuremath{\ifthenelse{\not \equal{#2}{} }  {#1(#2)}  {#1}}}
\newcommand{\st}[1]{\Xop{st}{#1}}
\newcommand{\ns}[1]{\Xop{ns}{#1}}

\newcommand{\EncAsim}[1]{\Xop{\tau_f}{#1}}

\newcommand{\ST}{\ensuremath{s_p}}
\newcommand{\FL}{\ensuremath{f_p}}

\newcommand{\XTNdom}{\ensuremath{\nsDom{\naturals}_+}}

\def\endsomething[#1]{
\ifmmode{\tag*{#1}}\else{\unskip\nobreak\hfil%
\penalty50\hskip1em\null\nobreak\hfil#1%
\parfillskip=0pt\finalhyphendemerits=0\endgraf}\fi}
\newcommand{\proofend}{\endsomething[$\blacksquare$]}

\newif\iflong
\longfalse

\begin{document}

\mainmatter  

\title{Non-null Infinitesimal Micro-steps:\\a Metric Temporal Logic Approach}


%
%
\author{ Luca Ferrucci, Dino Mandrioli, Angelo Morzenti, Matteo Rossi }
\institute{Politecnico di Milano, Dipartimento di Elettronica e Informazione \\ 
Piazza Leonardo da Vinci 32 -- 20133 Milano, Italy\\ 
\texttt{(ferrucci|mandrioli|morzenti|rossi)@elet.polimi.it}}

\maketitle

\begin{abstract}
Many systems include components interacting with each other that evolve with possibly very different speeds.
To deal with this situation many formal models adopt the abstraction of ``zero-time transitions'', which do not consume time.
These however have several drawbacks in terms of naturalness and logic consistency, as a system is modeled to be in different states at the same time.
We propose a novel approach that exploits concepts from non-standard analysis to introduce a notion of micro- and macro-steps in an extension of the TRIO metric temporal logic, called X-TRIO.
We use X-TRIO to provide a formal semantics and an automated verification technique to Stateflow-like notations used in the design of flexible manufacturing systems.
\end{abstract}

\bibliographystyle{splncs03}
\input{Introduction}

\input{SintaxSem}
\input{Encoding}
\input{Stateflow}
\input{Experiments}
\input{conclusion}

\subsubsection*{Acknowledgments.}
We would like to thank our colleagues at CNR-ITIA, Emanuele Carpanzano and Mauro Mazzolini, for providing expertise, insight and examples of design of FMS.

\bibliography{Bibliography}
\newpage

\input{Appendice}

\end{document}

%% file: Introduction.tex
\section{Introduction}
\label{sec:Int}

In many approaches to modeling time-dependent systems, each instant of a temporal domain $T$ is associated with exactly one ``state''.
%
This view can come into question when a system includes computational components that perform calculations whose durations are negligible with respect to the dominant dynamics of the system.
This occurs typically in embedded systems where some computing device, whose dynamics evolve at the pace of microseconds, monitors and controls an environment whose dynamics is in the order of the seconds. 
Imagine, for example, a controller of a reservoir that takes decisions on resource management in a few milliseconds, and actuates them in a few minutes.


A common abstraction adopted in literature to deal with this situation, one that is also widely accepted in the practice of systems development, consists in introducing a notion of ``zero-time transition'', where a state change occurs in such a short time that it can be neglected w.r.t. the other types of system evolution.
In this view, the system can traverse different states in zero time, thus a time instant $t$ can be associated with more than one state (e.g., the controller above could be in states ``update variables'' and ``make decision'' at the same time).
Examples of formalisms in which zero-time transitions are allowed are (see \cite{FMMR10}): timed Petri nets where transitions can have null firing times; some timed versions of Statecharts whose semantics is defined as a sequence of micro- and macro-steps, where only the latter ones advance time; various versions of timed or hybrid automata which separate transitions that produce a state change in null time from transitions that only make time progress.
In some sense these notations split time modeling in two separate domains: a logical domain, that orders events in terms of their logical precedence (e.g. the controller updates the variables before deciding whether to turn a switch on or off) and a physical domain over which the \textit{t} variable ranges.

The notion of zero-time transition, or micro-step, is a useful abstraction, but it inevitably entails some risks from the point of view of naturalness of modeling and safe mathematical analysis.
Not only the fact that a system can be in different states at the same time is counterintuitive from the standpoint of the traditional dynamical system view where the state is a function of time, it also exposes to the risk of contradictory assertions about system timing properties.
In \cite{GMM99} we proposed a natural way to overcome this difficulty through non-standard analysis (NSA): the temporal domain is extended by introducing infinitesimals,
i.e., numbers that are strictly less than any positive standard one.
We exploited this idea by replacing zero-time transitions with transitions that take a non-null, infinitesimal time in the context of our metric temporal logic language TRIO.


In this paper we further pursue our approach based on adopting a nonstandard time domain for TRIO to formalize micro- and macro-steps in dynamical systems.
The key novelty consists of introducing in TRIO the next-time operator typical of various temporal logics. 
Our approach retains the metric view of time that is typical of TRIO, but it avoids associating a fixed time distance to the next-time operator: the new state defined by it is entered after the current one at a time distance that can be a standard positive number, in the case of a macro-step, or an 
infinitesimal one in the case of a micro-step.
With this natural approach we preserve the intuitive concept that time and system state progress ``together'', but we also provide a mathematical foundation to support analysis and verification at different time scales. 

This extension of TRIO, called X-TRIO, allows us to describe in a natural way the formal semantics of --usually semi-formal-- notations that are widely used in industrial practice, in which zero-time transitions are a key concept.
In particular, in this paper we focus on the Stateflow notation \cite{SFLOW} that is common in the design of controllers of manufacturing systems.
Besides naturalness and generality, however, we pursue the goal of providing fully automated tools supporting the analysis of the modeled systems.
This is achieved by translating a decidable fragment of the X-TRIO logic, one that is expressive enough to fully capture the semantics of the target notation, into the Propositional Linear Temporal Logic with Both future and past operators (PLTLB) that is amenable to automated analysis by existing tools such as \zot{} \cite{ase08}.

\input{related}

This paper is structured as follows.
In Section \ref{sec:X-TRIO} we define the X-TRIO logic and study its relevant properties.
Then, in Section \ref{sec:Stateflow} we use X-TRIO to provide a formal semantics to the Stateflow notation, and we use the translation defined in Section \ref{sec:SemEncoding} to perform automated verification of an example of Flexible Manufacturing System.
Section \ref{sec:concl} concludes and hints at possible extensions and enhancements of this work.

%% file: related.tex
In the literature, other works \cite{BBCP12,BK09} have used NSA
to provide a formal and rigorous semantics to timing features of various kinds
of notations for system modeling. In \cite{BBCP12} NSA is used to describe a hybrid
system modeled in Simulink, in presence of cascaded mode changes. In \cite{BK09}, a
complete system theory is defined, adopting a theoretical approach to investigate
computability issues.


Since the introduction of Statechart (the language on which Stateflow is
based) several different semantics have been defined for it. The three most classical
ones, the fixpoint \cite{PS1991}, STATEMATE \cite{SCSEM}, and UML semantics, 
differ in the features adopted for step execution, and have been fully analyzed 
in \cite{Esh09}. In the present work we focus on Stateflow because of its widespread use in
industrial settings, but our approach is general enough to be adjusted to any of
the semantics defined for Statecharts or other state-based formalisms that use
the abstraction of micro- and macro-steps.


Notions of zero-time transitions, micro- and macro-steps appear very naturally 
when reasoning about computations of embedded systems, so, rather 
unsurprisingly, they arise in real-time temporal logics. Since the very early 
developments in this field, approaches were introduced that admit zero-time 
transitions at the price of associating multiple states to single time instants 
\cite{Ost89}. Our approach is akin to that of \cite{Koy90}, which introduces 
a general framework accomodating suitable time structures supporting the
notion of micro- and macro-steps, but does not address issues of decidability and verification.
The proposal in \cite{GH02DurCalc} provides notations for modeling micro-steps in the framework of 
Duration Calculus, which, unlike TRIO, is a logic based on intervals: it defines a decidable fragment of the
notation but does not give algorithms or build tools supporting verification.   
Other works are only partially connected to ours, as they deal with issues concerning the modeling and development of embedded systems at various time scales: \cite{HMP92DigitalClocks} and \cite{FR10-TOCL10} 
deal with issues 
of sampling and digitization, \cite{BH10} and \cite{CCM+91} discuss issues related with time granularity, and \cite{HQR99AssumeGuarantee} provides a refinement method based on assume-guarantee induction over 
different time scales. 



%% file: SintaxSem.tex
\section{The X-TRIO logic}
\label{sec:X-TRIO}

In this section we introduce the X-TRIO logic.
After some necessary background we define the syntax and semantics of the language.
Then, we study the relevant properties of the logic: we show the undecidability of X-TRIO in its general form, and we identify a subset whose satisfiability problem can be reduced to that of PLTLB, thus providing an effective mechanism to verify X-TRIO models.

\subsection{Background, syntax and semantics of X-TRIO}
\label{sec:Syntax}

The original TRIO language \cite{TRIO1} is a general-purpose specification language suitable for modeling real-time systems.
It is a temporal logic  supporting a metric on time.
TRIO formulae are built out of the usual first-order connectives, operators, and quantifiers, and the single basic modal operator, $\Dist{}$: for any formula $\phi$ and term $t$ indicating a time distance, the formula $\Dist{\phi,t}$ specifies that $\phi$ holds at a time instant whose distance is exactly $t$ time units from the current instant.
TRIO formulae can be interpreted both in discrete and dense time domains.

X-TRIO extends TRIO along two main lines.
First, the temporal domain $T$ is augmented with infinitesimal numbers (from the theory of non-standard analysis founded by A. Robinson \cite{NSA}):
intuitively, a number $\epsilon$ is infinitesimal if $\epsilon \geq 0$ and $\epsilon$ is smaller than any number in $T_{>0}$.
The original values of $T$ are classified as \emph{standard} and are characterized by predicate $st$; that is, $x$ is standard iff $\st{x}$ holds.
$T$ is augmented with infinitesimal numbers and all numbers resulting from adding and subtracting infinitesimal non-zero numbers to and from standard ones.
Predicate $\ns{x}$ denotes that $x$ is \emph{non-standard}; for each $x$, $\st{x}$ holds if and only if $\ns{x}$ does not hold.
Notice that 0 is the only infinitesimal standard number and that non-standard numbers are of the form $v \pm \epsilon$, where $\st{v}$ holds, and $\epsilon$ is infinitesimally greater than 0.
Then, NSA provides an axiomatization that allows one to apply all arithmetic operations and properties of traditional analysis in an intuitive way: for instance the sum of two standard numbers is standard, the sum of two infinitesimal numbers is an infinitesimal and the sum of an infinitesimal with a standard number is a non-standard number.
The theory of NSA introduces, in addition to the notion of infinitesimal numbers and operations on them, the notion of infinite numbers (which are, intuitively, greater than any value in $T$), plus a rich set of results that make NSA an appealing framework for reasoning on both familiar and new objects.
In this paper we exploit some of the terminology and concepts of NSA to provide an elegant characterization of zero-time steps, but we do not make use of the full power of the theory; for example, we do not deal with infinite numbers (i.e., we have that $\ns{x}$ iff $x = v \pm \epsilon$, with $\st{v}$ and $\epsilon$ infinitesimal), as they seem of little use when dealing with zero-time steps.

We assume $\reals$ as the original time domain $T$.
We denote the extension of $T$ with infinitesimal numbers as $\nsDom{T}$.
$\nsDom{T}$ is a totally ordered set of numbers.Throughout the paper we focus on subsets of $\nsDom{\reals}$.
In particular, we will consider the $\nsDom{\naturals}$ domain of naturals augmented with infinitesimal numbers.

The second major novelty of X-TRIO is the introduction of the \emph{next} operator $\X{}{}$ which is typically used to describe the evolution of dynamical systems as a sequence of discrete steps.
Unlike the traditional use of the operator in a metric setting, however, the time distance between two consecutive states is not implicitly assumed as a time unit; on the contrary it can be any standard or non-standard positive number.
Precisely, we introduce two different types of $\X{}{}$ operator, namely $\X{st}{}$ and $\X{ns}{}$.
Intuitively, the formula $\X{st}{\phi}$ is true in the current instant iff $\phi$ is true in the next state entered by the system and this occurs at a time instant that is a standard number; conversely, formula $\X{ns}{\phi}$ is true iff in the next state, $\phi$ is true and the occurrence time is a non-standard number.
We will use these two operators to distinguish between two typical ways of modeling system evolution: $\X{st}{}$ will formalize \emph{macro-steps} i.e. transitions that "consume real, tangible time", whereas $\X{ns}{}$ will describe \emph{micro-steps} which are often formalized as zero-time transitions.
\emph{Yesterday} operators $\Y{st}{}$ and $\Y{ns}{}$ are introduced in a similar manner.

The syntax of X-TRIO is defined as follows:
\vspace{-0.3cm}
\begin{align*}
&\phi\ :=\, p\, |\, \neg \phi\, |\, \phi_1 \wedge  \phi_2\, |\, \Dist{\phi,k}\, |\, \X{st}{\phi}\, |\, \X{ns}{\phi}\, |\, \Y{st}{\phi}\, |\, \Y{ns}{\phi}\, |\, \forall t.\tau \\
&\tau\ :=\, \phi\, |\, \Dist{\phi,t}\, |\, t = k\, |\, t < k\, |\, \tau_1 \wedge \tau_2\, |\, \neg \tau 
\end{align*}
\vspace{-1cm}

For the purposes of this paper we restrict the set of atomic propositions $AP$ to propositional variables $p$, 
and the set $V$ of temporal terms to variables $t$ and constants $k$.
Temporal terms $t$ take values in the time domain $\nsDom{T}$ and can appear only in closed formulae.
We leave first-order extensions of the logic to future work.
Symbols $\true$, $\false$, $\lor$, $\impl$, $\exists$, etc. are derived as usual.
We introduce the derived operators of X-TRIO in the same way as in TRIO.
The derived temporal operators used in this paper are shown in Table \ref{tab:XTrioOp}.
%
\begin{table}[b]
\begin{center}
\begin{tabular}{cccc} 
\hline\noalign{\smallskip}
OPERATOR & DEFINITION \\
\noalign{\smallskip}\hline\noalign{\smallskip}
\AlwF{}{\phi} & $\forall d (d \geq 0 \rightarrow \Dist{\phi,d})$ \\
\SomF{}{\phi} & $\exists d (d \geq 0 \land \Dist{\phi,d})$ \\
\WithinF{}{\phi, \delta} & $\exists d (0 \leq d \leq \delta \land \Dist{\phi,d})$ \\
\Until{}{\phi,\psi} & $\exists d \geq 0 (\Dist{\psi,d} \land \forall v (0 \leq v < d \rightarrow \Dist{\phi,v}))$ \\
\Since{}{\phi,\psi} & $\exists d \geq 0 (\Dist{\psi,-d} \land \forall v (-d < v \leq 0 \rightarrow \Dist{\phi,v}))$ \\
\noalign{\smallskip}\hline
\end{tabular}
\end{center}
\caption{X-TRIO derived temporal operators.}
\label{tab:XTrioOp}       
\end{table}

 A model-theoretic semantics for X-TRIO is defined by following a fairly standard path on the basis of a temporal structure $S= \langle \nsDom{T}, \beta, \nu, \sigma \rangle$, where:
\vspace{-1cm}
\begin{itemize}
	\item $\nsDom{T}$ is the time domain such that $\forall t \in \nsDom{T}$ it is $t \geq 0$.
	\item $\beta:\nsDom{T}\longrightarrow 2^{AP}$ is an interpretation function that associates each instant of time $t$ with the set of atomic propositions $\beta(t)$ that are true in $t$.
	\item $\nu:V\longrightarrow \nsDom{T}$ is an evaluation function that associates with each temporal term of the set $V$ a value in $\nsDom{T}$.
	\item $\sigma = \{ \sigma_i | i \in \naturals: \sigma_i \in \nsDom{T} \wedge \sigma_0 = 0 \wedge \forall j \in \naturals(j < i \Rightarrow \sigma_j < \sigma_i) \wedge \forall t \in \nsDom{T}(\sigma_i<t<\sigma_{i+1} \Rightarrow \beta(\sigma_i) = \beta(t))\}$ is the distinguishing element of X-TRIO temporal structure; it is a (possibly infinite) sequence of time instants starting from the initial instant 0, called \emph{History}.
	Intuitively, it represents the discrete sequence of instants when the system changes state; thus, the $\X{}{}$ operator represents a step moving from $\sigma_i$ to $\sigma_{i+1}$.
\end{itemize}
Then the satisfaction relation $\vDash$  of an X-TRIO formula $\phi$ by structure $S= \langle \nsDom{T}, \beta, \nu, \sigma \rangle$ at a time instant $i \in \nsDom{T}$ is defined as follows:

{\small
\begin{align*}
&S,i \vDash p \textrm{ iff } p \in \beta(i) \\
&S,i \vDash \neg \phi \textrm{ iff } S,i \nvDash \phi \\
&S,i \vDash \phi_1 \wedge \phi_2 \textrm{ iff } S,i \vDash \phi_1 \textrm{ and } S,i \vDash \phi_2 \\ 
&S,i \vDash \Dist{\phi,k} \textrm{ iff }  i+\nu(k) \in \nsDom{T} \textrm{ and } S,i+\nu(k) \vDash \phi & \\
&S,i \vDash \Dist{\phi,t} \textrm{ iff } i+\nu(t) \in \nsDom{T} \textrm{ and } S,i+\nu(t) \vDash \phi \\
&S,i \vDash \X{st}{\phi} \textrm{ iff there is } j\in \naturals \textrm{ s.t. } \sigma_j \leq i < \sigma_{j+1}, \st{\sigma_{j+1}} \textrm{ and } S,\sigma_{j+1} \vDash \phi \\
&S,i \vDash \X{ns}{\phi} \textrm{ iff there is } j\in \naturals \textrm{ s.t. } \sigma_j \leq i < \sigma_{j+1}, \ns{\sigma_{j+1}} \textrm{ and } S,\sigma_{j+1} \vDash \phi \\
&S,i \vDash \Y{st}{\phi} \textrm{ iff there is } j\in \naturals \textrm{ s.t. } \sigma_{j-1} < i \leq \sigma_{j}, j>0,  \st{\sigma_{j-1}} \textrm{ and }  S,\sigma_{j-1} \vDash \phi \\
&S,i \vDash \Y{ns}{\phi} \textrm{ iff there is } j\in \naturals \textrm{ s.t. } \sigma_{j-1} < i \leq \sigma_{j}, j>0, \ns{\sigma_{j-1}} \textrm{ and } S,\sigma_{j-1} \vDash \phi \\
&S,i \vDash \forall d.\phi  \textrm{ iff for all } \nu' \textrm{ that differ from } \nu \textrm{ at most for } d, \langle \nsDom{T}, \beta, \nu', \sigma \rangle,i \vDash \phi
\end{align*}
}
%
A formula $\phi$ is \emph{satisfiable} in a structure $S= \langle \nsDom{T}, \beta, \nu, \sigma \rangle$ when $S,0 \vDash \phi$.

In the rest of the paper, we focus our attention on a fragment of X-TRIO, which we name $\XTN$, that is sufficiently expressive for the purpose of providing Stateflow with a formal semantics and that is, under suitable conditions, decidable.
$\XTN$ formulae are interpreted on the temporal domain $\XTNdom \subset \nsDom{\reals}$ which includes exactly all numbers of the form $v+k\epsilon$, where $v,k \in \naturals$ and $\epsilon > 0$ is an infinitesimal \emph{constant} number fixed \emph{a priori}.
Thus, in $\XTNdom$, standard numbers are identified by the coefficient $k=0$.
$\XTN$ corresponds to the following syntactic fragment of X-TRIO, where $\epsilon$ is a constant and $\NowST$ is an operator with no arguments that is described below:
\begin{align*}
\phi\ :=\, & p\, |\, \neg \phi\, |\, \phi_1 \wedge  \phi_2\, |\, \Dist{\phi,1}\, |\, \Dist{\phi,-1}\, |\, \Dist{\phi,\epsilon}\, |\\
           & \Until{}{\phi_1, \phi_2}\, |\, \Since{}{\phi_1, \phi_2}\, |\ \X{st}{\phi}\, |\, \X{ns}{\phi}\, |\, \NowST
\end{align*}
In this fragment, $\Dist{\phi, 1+\epsilon}$ is an abbreviation for $\Dist{\Dist{\phi, \epsilon}, 1}$, and also $\Dist{\phi, 2} = \Dist{\Dist{\phi, 1}, 1}$, $\Dist{\phi, 2\epsilon} = \Dist{\Dist{\phi, \epsilon}, \epsilon}$, and so on.
Notice that the $\Until{}{}$ and $\Since{}{}$ operators of Table \ref{tab:XTrioOp} are primitive in $\XTN$, and we have the usual abbreviations $\SomF{}{\phi} = \Until{}{\true,\phi}$ and $\AlwF{}{} = \neg \SomF{}{\neg \phi}$.
As the syntax of $\XTN$ does not allow for variables, its temporal structures become triples of the form $S= \langle \nsDom{T}, \beta, \sigma \rangle$.
To distinguish between standard and non-standard instants, $\XTN$ introduces operator $\NowST$ such that $S,i \vDash \NowST \textrm{ iff } \st{i}$.
%

The restrictions introduced in $\XTN$, however, are not enough to make it decidable.
In fact, the following holds.

\begin{theorem}
\label{th:undec}
The satisfiability problem of the $\XTN$ logic is undecidable.
\end{theorem}

The proof of Theorem \ref{th:undec}, which can be found in \ref{subsec:proof_undec}, is by reduction of the halting problem of the 2-counter machine.
%
In Section \ref{sec:encoding} we introduce a sufficient condition that makes $\XTN$ decidable, but still expressive enough for our purposes.

%

%% file: Encoding.tex
\subsection{A decidable fragment of X-TRIO and its encoding in PLTLB}
\label{sec:encoding}

In this section we show the decidability of $\XTN$, under suitable conditions, by reducing the satisfiability problem of $\XTN$ to that of PLTLB.
The encoding of the transformation has been implemented in the \zot{} satisfiability checker.

PLTLB extends classic LTL \cite{Sch02} with past operators; its syntax as it will be used in the rest of this paper is the following:
\begin{align*}
&\phi :=\, p |\, \neg \phi\, |\, \phi_1 \wedge \phi_2\, |\, \X{L}{\phi}\, |\, \Y{L}{\phi}\, |\, \Uinfix{L}{\phi_1}{\phi_2}\, |\, \Sinfix{L}{\phi_1}{\phi_2}
\end{align*}

The semantics of PLTLB is defined over discrete \emph{traces}, representing infinite evolutions over time of the modeled system.
A trace is an infinite word $\pi = \pi(0)\pi(1)\ldots$ over the finite alphabet $\Sigma=2^{AP}$, where each $\pi(i)$ represents the set of atomic propositions that are true in $i$.
$\pi^i$ denotes the suffix of $\pi$ starting from $\pi(i)$.
We denote the satisfiability relation of PLTLB with $\satL$.
The definition of $\satL$ is straightforward if one considers that, for any $\phi$, $\Y{L}{\phi}$ is false at 0 \cite{Sch02}.

As a first step to encode $\XTN$ into PLTLB, we restrict histories $\sigma$ according to the following constraints:
\begin{enumerate}[{C}1.,ref={C}\arabic*]

\item
\label{enum:c1}
Either all standard natural numbers, or a bounded interval thereof including 0 belong to $\sigma$.

\item
\label{enum:c2}
If $\sigma_{i+1}$ is non-standard ($\ns{\sigma_{i+1}}$), then $\sigma_{i+1} - \sigma_{i} = \epsilon$.

\end{enumerate}
%
%
These constraints are not strictly necessary to obtain decidability, but they are not overly restrictive and they simplify the encoding for our purposes.
Notice also that, if $\sigma_{i+1}$ is standard ($\st{\sigma_{i+1}}$), then between $\sigma_i$ and $\sigma_{i+1}$ there is an infinite sequence of nonstandard numbers $\sigma_{i} + \epsilon, \sigma_{i} + 2\epsilon, \ldots$ such that, for all $k \in \naturals$, $\beta(\sigma_{i} + k\epsilon) = \beta(\sigma_{i})$.

To reduce the satisfiability problem of $\XTN$ (which is in general undecidable) to that of PLTLB (which is decidable), we need to apply further restrictions to the former. The key to obtain decidability is to make the evaluation of this operator meaningful only in standard instants.
To this purpose, we use the operator $\NowST$, that evaluates to true only in standard instants. 
To simplify the encoding further with a limited cost in expressiveness, we also impose that the value of formulae is meaningful only in instants that are "covered" by the history $\sigma$.
In fact, by definition of $\sigma$ in Section \ref{sec:X-TRIO}, there can be instants $t \in \nsDom{T}$ such that, for all $i$, $\sigma_i < t$.
In this case, $\sigma$ shows a classic \emph{Zeno behavior}, where it accumulates at a finite instant, signaling a model that changes state infinitely often in a finite interval.
Then, by convention, we state that formulae that are evaluated after one such accumulation point are false.
This can be achieved by considering every subformula $\psi$ of an $\XTN$ formula $\phi$ as an abbreviation for $\psi \land \SomF{}{\X{st}{\true} \lor \X{ns}{\true}}$.

The basic idea of the encoding is, given an $\XTN$ formula $\phi$, to build a corresponding PLTLB formula $\EncAsim{\phi}$ such that each model $S= \langle \XTNdom, \beta, \sigma \rangle$ of $\phi$ corresponds to a trace $\pi$ that is a model of $\EncAsim{\phi}$ such that every $\sigma_i$ maps to an element $j$ of $\pi$ where $\beta(\sigma_i) = \pi(j)$.
Then, we represent the transition $\sigma_i \longmapsto \sigma_{i+1}$ through the operator $\X{L}{}$.
Constraints \ref{enum:c1} and \ref{enum:c2} guarantee that the difference between $\sigma_{i+1}$ and $\sigma_i = v + k\epsilon$ is either $1-k\epsilon$ or $\epsilon$, depending on whether $\sigma_{i+1}$ is standard or not.
The encoding "flattens" the history $\sigma$ over $\pi$: to distinguish between standard and non-standard instants, we introduce a PLTLB propositional letter $\ST$ that marks elements of trace $\pi$ that correspond to a standard instants.
We also need to introduce a ``filling'' element in $\pi$ whenever in $\sigma$ there are two elements $\sigma_i$, $\sigma_{i+1}$ that are both standards, i.e. between two elements in $\pi$ that are marked as $\ST$ (see the proof  of Theorem \ref{th:correctEncAsim} in Appendix \ref{subsec:proof_correctEncAsim} for more details).
Filling elements are marked in $\pi$ through predicate $\FL$.



The translation schema $\EncAsim{}$ of Table \ref{tab:enc} transforms an $\XTN$ formula $\phi$ into an equally satisfiable PLTLB formula $\phi_L$.
\begin{table}[tb]
$$
\begin{array}{l}
  \begin{array}{ll}
  \EncAsim{p} = p_L \qquad \qquad \qquad \EncAsim{\NowST} = \ST \qquad \qquad & \EncAsim{\Dist{\phi,0}} = \EncAsim{\phi} \\
  \EncAsim{\neg \phi} = \neg \EncAsim{\phi}  & \EncAsim{\phi_1 \wedge \phi_2} = \EncAsim{\phi_1} \wedge \EncAsim{\phi_2}  
  \end{array}\\
  \begin{array}{ll}
  \EncAsim{\X{ns}{\phi}} = \X{L}{\EncAsim{\phi} \wedge \neg \ST} \qquad &  \EncAsim{\X{st}{\phi}} = \X{L}{(\EncAsim{\phi} \wedge \ST) \lor (\FL \land \X{L}{\EncAsim{\phi}})}
  \end{array}\\
  \begin{array}{l}
  \EncAsim{\Dist{\phi,\epsilon}} = \X{L}{\EncAsim{\phi} \wedge \neg \ST} \vee (\X{L}{\ST} \wedge \EncAsim{\phi}) \\
  \EncAsim{\Dist{\phi,1}} = \X{L}{\Uinfix{L}{\neg \ST}{(\EncAsim{\phi} \wedge \ST)}} \\
  \EncAsim{\Dist{\phi,-1}} = \ST \land \Y{L}{\Sinfix{L}{\neg \ST}{(\ST \land \EncAsim{\phi})}} \\
  \EncAsim{\Until{}{\phi,\psi}} = \Uinfix{L}{\EncAsim{\phi}}{\EncAsim{\psi}} \\
  \EncAsim{\Since{}{\phi,\psi}} = \Sinfix{L}{\EncAsim{\phi}}{(\X{L}{\neg \ST} \land \EncAsim{\psi})}\ \lor \Sinfix{L}{\EncAsim{\phi}}{(\X{L}{\ST} \land \EncAsim{\phi} \land \EncAsim{\psi})} \\
  \end{array}
\end{array}
$$
\caption{Translation schema $\EncAsim{}$.}
\label{tab:enc}
\end{table}

%
Schema $\EncAsim{}$ is completed by the assertions (A1)

\begin{equation}
\begin{array}{ll}
\ST \land \G{L}{(\ST \impl \X{L}{\FL \lor \neg \ST}) \land (\FL \impl (\Y{L}{\ST} \land \neg \ST \land \X{L}{\ST}))}
\end{array}
\end{equation}

which imposes that predicate $\ST$ holds in $\pi(0)$ and that $\FL$ always appears between two consecutive $\ST$, and (A2) 

\begin{equation}
\begin{array}{ll}
\G{L}{\FL \impl (\bigwedge_{p \in AP} p \dimpl \Y{L}{p})}
\end{array}
\end{equation}

which states that propositions do not change values between two standard instants $\sigma_i$ and $\sigma_{i+1}$.
The following result holds (see Section \ref{subsec:proof_correctEncAsim} for the proof).

\begin{theorem}
\label{th:correctEncAsim}
Given an $\XTN$ formula $\phi$, there is a structure $S=\langle \XTNdom, \beta, \sigma \rangle$ such that $S,0 \vDash \phi$ iff there exists a trace $\pi$ such that $\pi\ \satL \EncAsim{\phi} \land (\mathrm{A1}) \land (\mathrm{A2})$.
\end{theorem}

From translation schema $\EncAsim{}$ and Theorem \ref{th:correctEncAsim} we can prove the following.

\begin{theorem}
The satisfiability problem for  $\XTN$ as restricted in this section is decidable and PSPACE-complete.
\end{theorem}

%% file: Stateflow.tex
\section{Exploiting X-TRIO to analyze Stateflow diagrams}
\label{sec:Stateflow}

In this section we present an application of $\XTN$ to provide the Stateflow notation with a formal semantics that includes a precise, metric notion of time; this allows us 
to introduce metric constraints in the notation and
to formally analyze real-time requirements and properties.
We exploit the $\XTN$-based semantics of Stateflow to perform automated formal verification of some properties of interest of the controller of a Flexible Manufacturing System (FMS), which is used in the section as an example to illustrate the Stateflow notation.

\subsection{Stateflow diagrams and their semantics in $\XTN$}
\label{sec:SemEncoding}

The Stateflow notation is a variation of Statecharts; it describes finite state machines performing discrete transitions between states in a simple and intuitive way.
In a nutshell, a Stateflow diagram is composed of: (i) a finite set of typed variables $V$ partitioned into input ($V_I$), output ($V_O$), and local ($V_L$) variables; input and output events are represented, respectively, through Boolean variables of $V_I$ and $V_O$; (ii) a finite set of states $S$ which can be associated with \emph{entry}, \emph{exit} and \emph{during} actions, which are executed, respectively, when the state is entered, exited, or throughout the permanence of the system in the state; (iii) a finite set of transitions, $H$, that may include guards (i.e., constraints) on the variables of $V$ and actions.
An \emph{action} is the assignment of the value of an expression over constants and variables of $V$ to a non-input variable.
We assume all variables in $V$ to take values in a finite domain, which we represent by $D_V$.

We illustrate the notation through the example of a robotic cell composed of a robot arm that loads and unloads various parts on two machines, $M_1$ and $M_2$.
The cell is served by a conveyor belt, which provides pallets to be processed.
There are two types of pallets, $A$ and $B$, which are precessed, respectively, by machine $M_1$ and by machine $M_2$.
After processing, the finished parts are discharged from the cell by means of the conveyor out belt. Figure \ref{fig:Robot} shows a Stateflow diagram describing the behavior of the robot arm.

\begin{figure*}[hbtf]
\centering
\includegraphics[height=0.55\textheight]{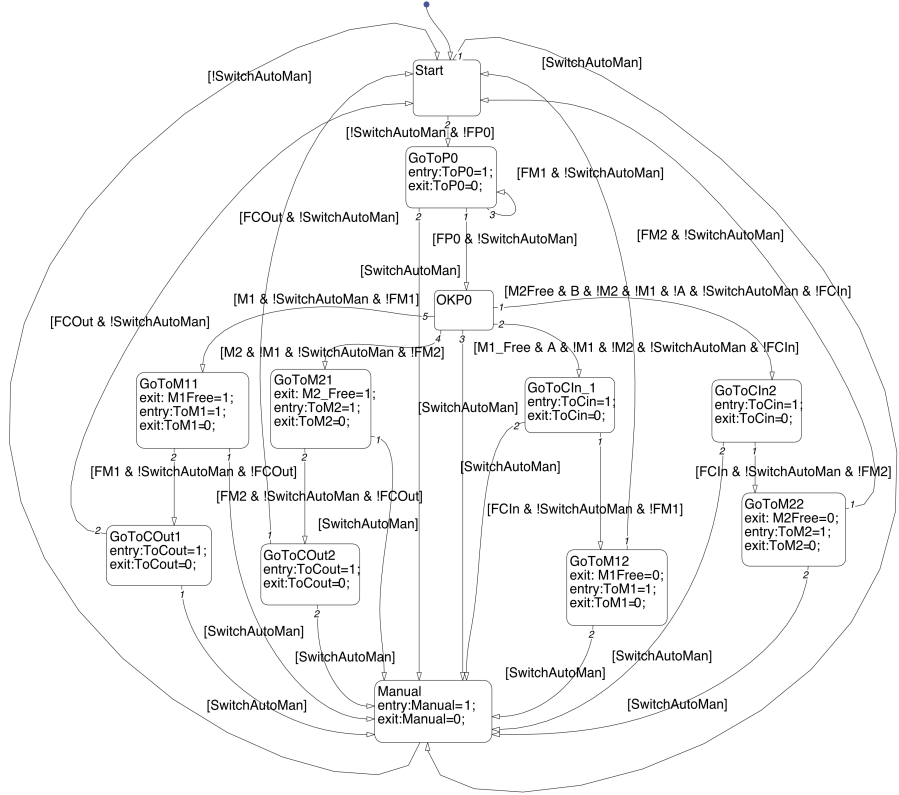}
\caption{Stateflow diagram of the controller of the robotic arm}
\label{fig:Robot}       
\end{figure*}

At any time, the robot arm can switch from automatic to manual mode or from manual to automatic mode upon a suitable command from the operator. For example, in the graph of Figure \ref{fig:Robot}, the transition from state \emph{GoToP0} to state \emph{OKP0} is enabled when a photocell signals that the robot arm has reached the central position \emph{P0}, setting the input variable \emph{FP0}.

\cite{SFLOW} presents the complete, informal, specification of Stateflow diagrams, but it does not provide a precise definition of their semantics. Our one is based on the STATEMATE semantics of Statecharts \cite{SCSEM}.

Stateflow semantics hinges on the concept of \emph{run}, which represents the reaction of the system to a sequence of input events.
A run is a sequence of \emph{configurations}; each configuration $\langle s, \nu \rangle$ pairs the current state $s \in S$ with an evaluation function $\nu : V \rightarrow D_V$ representing the current values of the variables.
The configuration changes only when an enabled transition is executed.
An enabled transition must be executed, which entails that a Stateflow model must be internally deterministic.
Input events, however, occur in a nondeterministic manner, so the model overall is nondeterministic.

The semantics of time evolution in Statecharts/Stateflow diagrams has proven difficult to pin down precisely, and different solutions have been proposed in the literature (e.g., \cite{RM}).
Our model is of the so-called \emph{run-to-completion} variety.
In this model the system reacts to the input events by performing a sequence of reactions (\emph{macro-steps}).
Within every macro-step, a maximal set of enabled transitions (\emph{micro-steps}) is selected and executed based on the events generated in the previous macro-step.
Micro-steps are executed infinitely fast, with time advancing only at macro-step boundaries, when the system reaches a \emph{stable} configuration, i.e., in which no transition is enabled.
In other words, micro-steps take zero time to execute; when no transition is enabled, time advances and the configuration changes when a new input event is received from the environment. As for STATEMATE, components sense input events and data only at the beginning of macro-steps and communicate output events and data only at their end. 
In the semantics outlined above each run identifies a sequence of time instants $\{t_i\}_{i \in \naturals}$, one for each macro-step, hence the time domain is discrete.
This is consistent with the underlying physical model of our test case, as the PLCs on which FMS control solutions are built are governed by discrete clocks.
In a sense, each macro-step corresponds to a \emph{clock cycle} of the modeled PLC. 

For example, if, at the beginning of a macro-step, the robot arm of Figure \ref{fig:Robot} is in position \emph{P0} (i.e., in state \emph{OKP0}) and a pallet of type $A$ is to be delivered to machine $M_1$, the transition between states \emph{OKP0} and \emph{GoToCIn1} is enabled, so the robot arm executes a micro-step and the output variable \emph{ToCIn} is set to true.
At this point, the whole system has reached a stable state, since the robot arm must wait for machine $M_1$ to terminate processing the pallet.
The termination event is modeled by setting the input variable \emph{FCIn} to true.

We now formalize the semantics of Stateflow diagrams through $\XTN$ formulae.
As the domain $D_V$ of Stateflow variables is assumed to be finite, it can be represented through a set of propositional letters: given a variable $v \in V$ and a value $k \in D_V$ when $v_k$ is true this represents that the value of $v$ is $k$.
Similarly for the state space $S$.
For readability, we write $v = k$ instead of $v_k$.

For each Stateflow transition $H_i : {s_i} \stackrel{g_i/a_i}{\rightarrow} {s'_i}$ from state $s_i$  state $s'_i$ with guard $g_i$ and action $a_i$, we introduce the following formula:
\begin{equation}
\begin{array}{ll}
\label{eq:8}
\AlwF{}{(\gamma_i \wedge s = s_i) \impl \X{ns}{s = s'_i} \wedge \alpha_i \wedge \alpha_{ex_{s_i}} \wedge \alpha_{en_{s'_i}}}
\end{array} 
\end{equation}
where $\gamma_i$ is an $\XTN$ formula encoding guard $g_i$, and $\alpha_i$, $\alpha_{ex_{s_i}}$ and $\alpha_{en_{s'_i}}$ are $\XTN$ formulae encoding, respectively, the transition action $a_i$, and the \emph{entry} and \emph{exit} actions of states ${s_i}$ and ${s'_i}$.
Formula \eqref{eq:8} formalizes the execution of a micro-step: it asserts that if the current state is $s_i$ and the transition condition $\gamma_i$ holds in the current configuration, then in the next micro-step the active state must be $s'_i$ and the $entry$ actions of $s'_i$ and the $exit$ actions of $s'_i$ are executed. Thus, operator $\X{ns}{}$ replaces a zero-time transition. If no transition is enabled, the configuration does not change, which is captured by the following formula:
\begin{equation}
\begin{array}{ll}
\label{eq:9}
 \AlwF{}{\bigwedge_{i=1}^{|H|} \neg(\gamma_i \wedge s = s_i) \impl NOCHANGE} 
\end{array} 
\end{equation}
where subformula $NOCHANGE$, which is not further detailed for space reasons, asserts that in the next micro-step the current state and the values of all output and local variables do not change.

The time advancement of our semantics is modeled through operator $\X{st}{}$: every time the system reaches a stable state (where no transition is enabled), the time advances to the next standard number. This is captured by the formula:
\begin{equation}
\label{eq:timeadv}
 \AlwF{}{\bigwedge_{i=1}^{|H|} \neg(\gamma_i \wedge s = s_i) \dimpl \X{st}{\true}}.
\end{equation}

The complete definition of the behavior of the transitions of the Stateflow diagram is given by $\left(\bigwedge_{i=1}^{|H|} (\ref{eq:8})_i\right) \land (\ref{eq:9}) \land (\ref{eq:timeadv})$.

Finally, we introduce a formula asserting that input variables $V_I$ change values only at the beginning of a macro-step, i.e. when the system is in a standard instant of time.
In other words, if the next time instant is non-standard, then the values of the input variables must be the same as those in the current instant:

\begin{equation}
\begin{array}{ll}
\label{eq:10}
\AlwF{}{\X{ns}{\true} \impl (\bigwedge_{v \in V_I, x \in D_V} v=x \impl \X{ns}{v=x})}
\end{array}
\end{equation}

The formula $SYS$ encoding the behavior of the overall system is given by the conjunction of formulae $\bigwedge_{i=1}^{|H|} (\ref{eq:8})_i$, (\ref{eq:9}-\ref{eq:10}), plus others not shown for brevity.
Formula $SYS$ characterizes precisely the runs of the corresponding Stateflow diagram, that is, it holds exactly for the runs modeled through the diagram.

%% file: Experiments.tex
\subsection{System properties verification and experimental results}
\label{sec:experiments}

The formalization introduced in Section \ref{sec:SemEncoding} has been implemented in the \zot{} tool to perform the verification of some typical real-time properties of the example FMS system.
\zot{} \cite{ase08} is a bounded satisfiability checker which supports the verification of PLTLB models.
It solves satisfiability (and validity) problems for PLTLB formulae by exploiting Satisfiability Modulo Theories (SMT) \cite{SMT} solvers.
Through \zot{} one can check whether stated properties hold for the system being analyzed (or parts thereof) or not; if a property does not hold, \zot{} produces a counterexample that violates it. 



As a first example, we check that the modeled system does not have Zeno runs, which would make it unrealizable. The system shows a Zeno behavior if, from a certain point on, ``real'' time does not advance, i.e., no macro-steps are performed.
The \emph{presence} of Zeno runs is formalized as follows:
\begin{equation} 
\begin{array}{ll}
\label{eq:zeno_runs}
\SomF{}{\AlwF{}{\X{ns}{\true}}}
\end{array}
\end{equation}

Formula (\ref{eq:zeno_runs}) states that, from a certain instant on, the clock does not tick any more, i.e. the trace presents an infinite sequence of non-standard instants.
We checked through the \zot{} tool that formula $SYS \land (\ref{eq:zeno_runs})$ is \emph{unsatisfiable}, hence no runs of the system show property (\ref{eq:zeno_runs}), and the system is devoid of Zeno runs. 


Through $\XTN$ it is possible to formalize different variations for the intuitive notion of ``until'', for example one that takes into account only the last micro-step of each macro-step, i.e. when the system reaches a ``stable state''.
Informally, $\Until{stable}{\phi,\psi}$ holds if there is a future macro-step such that in its last micro-step $\psi$ holds, and $\phi$ holds in the last micro-step of all macro-steps before that.
The $\Until{stable}{}$ operator is useful to check properties that predicate only over the ``real'' time.
It is defined by the following $\XTN$ formula:
\begin{equation} 
\label{eq:until_stable}
\Until{stable}{\phi,\psi} \defeq \Until{}{\X{st}{\true} \impl \phi, \X{st}{\true} \wedge \psi}
\end{equation}
where the last micro-step is identified by the fact that its next instant is standard.
Another possible variant of ``until'', for example, is one that predicates only over the first instants of macro-steps, i.e., standard instants.
It is defined by the following $\XTN$ formula which exploits predicate $\NowST$ of Section \ref{sec:encoding}:
\begin{equation} 
\label{eq:until_st}
\Until{st}{\phi,\psi} \defeq \Until{}{\NowST \impl \phi, \NowST \wedge \psi}
\end{equation}

We use operator $\Until{stable}{}$ to check for the existence of deadlocks in a system of synchronously evolving modules.
Our notion of deadlock is defined over macro-steps only, since we consider micro-steps to be \emph{transient} states that are non-observable outside of a module. Then, we say that the system is in deadlock if \emph{all} of its components are in a deadlock state. If $E$ is the set of components of the system, where each $e \in E$ is described through a Stateflow diagram with state space $S_e$, the following $\XTN$ formula captures this notion of deadlock: :
\begin{equation}
\label{eq:deadlock}
\bigwedge_{e \in E} \bigvee_{x \in S_e} \SomF{stable}{\AlwF{stable}{s_e = x}}
\end{equation}
where $\SomF{stable}{\phi}$ and $\AlwF{stable}{\phi}$ are, as usual, abbreviations for\linebreak
$\Until{stable}{\true,\phi}$ and $\lnot \SomF{stable}{\lnot \phi}$, respectively.

The last property we present in this paper is a real-time  property that states whether it is possible to produce and deliver one processed workpiece of any kind within $L$ time units from the system startup.
The property is captured by the following formula, with the obvious meaning of the $\WithinF{stable}{}$ operator:
\begin{equation}
\begin{array}{ll}
\label{eq:reachability}
 \WithinF{stable}{(s_{Rob} = GoToCo1) \vee (s_{Rob} = GoToCo2) , L}
\end{array}
\end{equation} 

The formula checks whether, within a time $L$ from the system startup, one of the states \emph{GoToCo1} or \emph{GoToCo2} of Figure \ref{fig:Robot} is reachable. The Stateflow diagram reaches state \emph{GoToCo1} if a workpiece of any type has been produced by machine $M_1$, similarly for the other. By testing various values for $L$, we found that the 
minimum $L$ for which formula (\ref{eq:reachability}) holds is 16. 


Performance results obtained during the verification of properties above are shown in Table \ref{tb:margins}.
Verifications was performed with a bound of 70 time units, which is a user-defined parameter that corresponds to the maximal length of runs analyzed by \zot{}.
The table shows the time spent to check the property, the memory occupation and the result, i.e. whether the property holds or not.\footnote{All tests have been performed on a 3.3GHz QuadCore PC with Windows 7 and 4GB of RAM. The verification engine used was the SMT-based \zot{} plugin of \cite{SMT}; the solver was z3 3.2 (\url{http://research.microsoft.com/en-us/um/redmond/projects/z3/}).}

\begin{table}[!bt]
\begin{center}
\begin{tabular}{cccc}
Formula & Time (sec) & Memory (Mb) & Result \\\hline
Zeno Paths detection (\ref{eq:zeno_runs}) & 85 &  264 & No\\
Deadlock detection (\ref{eq:deadlock}) & 17991 & 268 & No \\
Workpiece, L=15 (\ref{eq:reachability}) & 407 & 260 & No\\
Workpiece, L=20 (\ref{eq:reachability}) & 89 & 272 & Yes
\end{tabular}
\end{center}
\caption{Test results}
\label{tb:margins}
\end{table}

Considering that the sole Stateflow diagram of the controller of the robot arm of Figure \ref{fig:Robot} has $12 \cdotp 2^{18}$ possible configurations, i.e., $|S| \cdotp 2^{|D_V|}$, the first verification experiments are encouraging, and show the feasibility of the approach. In fact, we were able to detect deadlocks in an early specification of the FMS that stemmed from an incorrect communication protocol between the robot and machine $M_1$.


%% file: conclusion.tex
\section{Conclusions and Future Work}
\label{sec:concl}

We introduced a novel approach to the modeling and analysis of systems that evolve through a sequence of micro- and macro-steps occurring at different time scales, such that the duration of the micro-steps is negligible with respect to that of the macro-steps. In some sense, we can position our approach in between the "time granularity approach" \cite{CCM+91} where different but positive standard and comparable time scales are adopted at different levels of abstraction and the "zero-time transition" approach \cite{SCSEM}, \cite{Ost89} which instead "collapses" the duration of some action to a full zero. By introducing the notion of infinitesimal duration for micro-steps and by borrowing the elegant notation of NSA to formalize them, we overtake the limitations of the two other cases and generalize them: on the one side, unlike traditional mappings of different but positive standard time granularities, infinitesimal steps may accumulate in unbounded or unpredictable way, thus allowing for the analysis of usually pathological cases such as zeno behaviors; on the other side by imposing that the effect of an event strictly follows in time its cause, we are closer to the traditional view of dynamical system theory, and we can reason explicitly about possible synchronization between different components even at the level of micro-steps.

We pursued our approach through the novel language X-TRIO, which includes both metric operators on continuous time and the next-time 
operator to refer to the next discrete state in the computation. Under simple and realistic conditions X-TRIO can be coded into an equivalent PLTLB formulation, which makes it amenable to automatic verification. 

To demonstrate the usefulness of our approach we developed a case study where we applied X-TRIO to formalize the semantics of Stateflow, to specify through it a simple robotic cell, and to prove a few basic properties thereof.

We emphasize the generality and flexibility of our approach. Although in this paper we focused essentially on its application to formalizing (one particular semantics of) the Stateflow notation, it should be already apparent that the same path could be followed for different operational and descriptive notations and for their semantic variations. For instance, notice how we came  up in a flexible way with simple formalizations of different interpretations of the $\Until{}{}$ operator; still others could be easily devised according to the needs of different applications.

Such a generality will be pursued along several dimensions. The present choice of just one time unit for micro- and one for macro-steps is good enough for Stateflow and FMS but is not a necessary restriction: different, fixed or even variable durations for micro-steps could be used to model different components of a global system and their synchronization at the micro-level; macro-steps too could have different durations. On the other hand, non-zero infinitesimal durations for micro-steps are particularly well-suited to investigate -the risk of- dangerous behaviors such as zenoness; however, once such a pathological property has been excluded it could be useful to turn back to a finite metric of micro-steps, perhaps exploiting different time granularities: something similar occurs during hardware design where, in various contexts, the designer analyses the risk of critical races and the duration of precise finite sequences of micro-steps,  or "collapses" all such sequences in an "abstract zero-time".
Our approach allows the designer to manage all such "phases" in a uniform an general way.

Another dimension along which it is worth exploiting the generality of our approach is the issue of decidability. The trade-off between expressive power and decidability (efficiency) offers many opportunities. 
Other, more general, versions of X-TRIO possibly supported by decision algorithms different from, or complementary to, the translation into PLTLB are under investigation.



%% file: Appendice.tex
\appendix
\section{Theorem proofs}
\label{sec:proofs}

\subsection{Proof of Theorem \ref{th:undec}}
\label{subsec:proof_undec}

\begin{proof}

To demonstrate theorem \ref{th:undec} we reduce the halting problem of a 2-counter machine to the satisfiability problem of $\XTN$ formulae.
To achieve this, we define a set of $\XTN$ formulae that formalize the increment and decrement of the 2 counters.

More precisely, we associate one counter with the sequence of even standard numbers, and the other with the sequence of odd standard numbers, in the following way:

\begin{itemize}
	\item we associate two different propositional letters, $E$ and $O$, with each standard instant of $\sigma$ s.t. when the current standard instant is an even (resp. odd) integer number then only $E$ (resp. $O$) holds.
	They do not hold in non-standard instants.
	These constraints are represented by the following $\XTN$ formulae  (we show the case of even instants):
	
\begin{equation*}
\begin{array}{rl}
\label{eq:1}
E \Rightarrow \X{st}{O} \vee \X{ns}{\Until{}{\neg O \wedge \neg E, \neg \X{ns}{\true} \land \neg O \land \neg E}} \\
E \Leftrightarrow \Dist{O, 1} 
\end{array}
\end{equation*} 

Similarly for $O$.

	\item Given two consecutive standard instants $\sigma_j$ and $\sigma_i$ in $\sigma$ (i.e., such that where $\sigma_i = \sigma_j + 1$), there is a finite (possibly empty) sequence of non-standard instants between them since $\sigma$ is discrete.
	This finite sequence has length $\left|i-(j+1)\right|$.
	We indicate this subsequence of instants $\sigma_{[j, i)}$ (notice that we include in $\sigma_{[j, i)}$ standard number $\sigma_j$, but not standard number $\sigma_i$).
	We introduce suitable $\XTN$ formulae to constrain sequence $\sigma_{[j,i)}$ to be partitioned into two further subsequences in which, in each instant, propositional letter $A$ (resp. $B$) holds (in addition, $A$ and $B$ are mutually exclusive).
	We use letters $A$ and $B$ them to ``mark'' each instant in $\sigma_{[j,i)}$ as in the example of Figure \ref{fig:ProofDec}.
	The sequence of $B$s ends in the last non-standard instant of $\sigma_{[j,i)}$.
	The following $\XTN$ formulae (which exploit the fact that when $\Until{}{\phi,\psi}$ holds, $\phi$ must hold up to the instant \emph{before} $\psi$ holds) formalize the behavior above:
	
\begin{equation}
\begin{array}{ll}
\label{eq:2}
A \impl \Until{}{A \wedge \X{ns}{\true},B} \\
B \impl \Until{}{B,\neg \X{ns}{\true}} \\
A \dimpl \neg B 
\end{array}
\end{equation} 

\begin{figure}[htbf]
\centering
\includegraphics[height=.065\textheight]{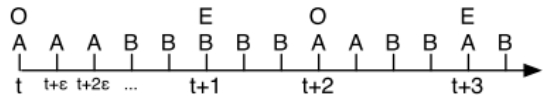}
\caption{Part of trace representing counters}
\label{fig:ProofDec}       
\end{figure}

\item We use the sequence of $A$ and $B$ to represent the two counters: the number of $A$'s starting from standard numbers marked with $E$ (resp. $O$) represent the first (resp. second) counter.
Then, we can encode the three operations \emph{increase}/\emph{decrease}/\emph{check if the counter is 0}, by manipulating the length of the sequence of $A$s in the following way (we show only the formulae of the counter of the even instants, it is similar for the other one):

 \begin{enumerate}
	\item The counter increases its current value if the sequence of $A$'s that starts at next even standard instant is such that the last $A$ of that sequence dists $2+\epsilon$ from the last $A$ of the current sequence of $A$'s.
	We can encode this condition through the following $\XTN$ formula:
	
\begin{equation*}
\begin{array}{rl}
	\label{eq:3}
		E \impl & (A \impl \Until{}{A, B \wedge \Dist{A \wedge \X{ns}{B},2}}) \\
		        & \land \\
		        & (B \impl \Dist{A \wedge \X{ns}{B},1})
\end{array}
\end{equation*} 
		
	\item The counter decreases its current value if, at the next even standard instant, the length of the sequence of $A$'s is shorter than the current one of exactly one $A$.
	We can encode this constraint through the following $\XTN$ formula:
	
\begin{equation*}
\begin{array}{rl}
	\label{eq:4}
		E \impl & (A \wedge \X{ns}{A} \impl \Until{}{A, A \land \X{ns}{A \land \X{ns}{B}} \wedge \Dist{A \land \X{ns}{B} ,2}}) \\
				& \land  \\
		        & (A \wedge \X{ns}{B} \impl \Dist{B,2})
\end{array}
\end{equation*} 
	
	The first formula describes the case where the current value of the counter is strictly greater than 1. The second formula instead describes the case where the current value of the counter is exactly 1. 
	
	\item The counter does not change its value if, at the next even standard instant, the length of the sequence of $A$'s is equal to the current one.
	We can encode this condition with following $\XTN$ formula:
	
\begin{equation*}
\begin{array}{rl}
	\label{eq:5}
		E \impl & (A \impl \Until{}{A, A \land \X{ns}{B} \wedge \Dist{A \wedge \X{ns}{B},2}} \\
				& \land \\
				& (B \impl \Dist{B,2})
	\end{array}
\end{equation*} 
	
		The first formula describes the case where the current value of the counter is strictly greater than 0. The second formula instead describes the case where the 			current value of the counter is exactly 0. 
	
	\item The counter is zero when the sequence of $A$'s is empty.
	In the case of the counter associated with even standard numbers we can encode this check with the following $\XTN$ formula:

\begin{equation*}
\begin{array}{rl}
	\label{eq:6}
		E \wedge B
	\end{array}
\end{equation*} 
			
 \end{enumerate}

\item Finally, at the initial instant of the sequence $\sigma$, which is an even number, $E$ holds and the corresponding counter value is 0. This is modeled by the following $\XTN$ formula evaluated at instant 0:

	\begin{equation}
	\begin{array}{ll}
	\label{eq:7}
		E \wedge B
	\end{array}
	\end{equation} 
	
\end{itemize}		

$\XTN$ formulae \eqref{eq:1}--\eqref{eq:7} formalize the core mechanisms of a 2-counter machine that can decide to increase/decrease or leave unchanged the values of the counters on the basis of the set of atomic propositions that are true in a given instant of time, which are used to represent the current state of the machine.
From this, the halting of the formalized machine can be expressed as a simple reachability of a final state.
Hence, we can conclude that the satisfiability problem of $\XTN$ is undecidable.
\proofend
\end{proof}

\subsection{Proof of Theorem \ref{th:correctEncAsim}}
\label{subsec:proof_correctEncAsim}

In order to prove Theorem \ref{th:correctEncAsim}, we first need to introduce two intermediate results.

\begin{lemma}
\label{lm:zeno}
Given an $\XTN$ formula $\phi$ in which all subformulae have the form $\psi \land \SomF{}{\X{st}{\true} \lor \X{ns}{\true}}$, and given two structures $S_1= \langle \XTNdom, \beta_1, \sigma \rangle$, $S_{2}= \langle \XTNdom, \beta_{2}, \sigma \rangle$ (i.e., which have the same history $\sigma$) such that, for all $t \in \XTNdom$ for which there is $i \in \naturals$ such that $t < \sigma_i$, it is $\beta_1(t) = \beta_2(t)$, then $S_1, 0 \models \phi$ iff $S_2, 0 \models \phi$.
\end{lemma}

\begin{proof}
We show a stronger result, from which Lemma \ref{lm:zeno} descends as corollary.
More precisely, we show that, given any $t \in \XTNdom$, $S_1, t \models \phi$ iff $S_2, t \models \phi$.
First of all, we remark that, if for each $t \in \XTNdom$ there is a $\sigma_i$ such that $t < \sigma_i$, then for
all $t \in \XTNdom$ it is $\beta_1(t) = \beta_2(t)$, hence the desired result.
In addition, notice that, in this case, condition $\SomF{}{\X{st}{\true} \lor \X{ns}{\true}}$ is true for all $t \in \XTNdom$, so the value of $\phi$ does not depend on it.

In the rest of the proof we consider the case in which there are instants $t$ such that, for all $i$, $\sigma_i < t$.
The set of such instants can be shown to have a minimum, which we indicate with $\overline{t}$, such that $\st{\overline{t}}$.
Then, history $\sigma$ accumulates at $\overline{t}$, and we separate two cases: $t < \overline{t}$ and $t \geq \overline{t}$.
In the case $t \geq \overline{t}$, $\SomF{}{\X{st}{\true} \lor \X{ns}{\true}}$ is false, hence for all $\phi$ both $S_1, t \not\models \phi$ and $S_2, t \not \models \phi$.
Then, we only need to consider the case $t < \overline{t}$.
The rest of the proof is by induction on the structure of $\phi$: consider a subformula $\psi$ of $\phi$.

If $\psi = p$, by hypothesis $\beta_1(t) = \beta_2(t)$; hence the result.

The cases $\psi = \neg \zeta$ and $\psi = \psi_1 \land \psi_2$ are trivial.

If $\psi = \Dist{\zeta, 1}$, then $S_1, t \models \psi$ iff $S_1, t+1 \models \zeta$, hence, by inductive hypothesis, iff $S_2, t+1 \models \zeta$, and iff $S_2, t \models \psi$.
Similarly for $\Dist{\zeta, -1}$ and $\Dist{\zeta, \epsilon}$.

If $\psi = \Until{}{\psi_1,\psi_2}$, $S_1, t \models \psi$ iff there is $t' \geq t$ such that $S_1, t' \models \psi_2$, and for all $t \leq t'' < t'$ it is $S_1, t \models \psi_1$; by inductive hypothesis this occurs iff $S_2, t' \models \psi_2$, and for all $t \leq t'' < t'$ it is $S_2, t \models \psi_1$, i.e., iff $S_2, t \models \psi$.
The case $\Since{}{\psi_1,\psi_2}$ is similar.

If $\psi = \X{st}{\zeta}$, then $S_1, t \models \psi$ iff there is $i \in \naturals$ such that $\st{\sigma_{i+1}}$, $\sigma_i < t \leq \sigma_{i+1}$ and $S_1, \sigma_{i+1} \models \zeta$; by inductive hypothesis this holds iff $S_2, \sigma_{i+1} \models \zeta$, hence the result.
Similarly for $\X{ns}{\zeta}$.
\proofend
\end{proof}

As a consequence of Lemma \ref{lm:zeno}, and also of the next result, given the restrictions introduced in Section \ref{sec:encoding}, in order to determine whether an $\XTN$ formula is satisfiable we need only focus on the sequence $\sigma$, and we can disregard the instants following an accumulation point, if any.

We introduce the following further intermediate result, in which we show that, in each interval $(\sigma_i, \sigma_{i+1})$ such that $\st{\sigma_{i+1}}$, the subformulae of $\phi$ have the same value in all $t \in [\sigma_i, \sigma_{i+1})$.

\begin{lemma}
\label{lm:ns_inv}
Given an $\XTN$ formula $\phi$ and a structure $S= \langle \XTNdom, \beta, \sigma \rangle$, if $\st{\sigma_{i+1}}$, then for any two instants $j, k \in \XTNdom$ such that $\ns{j}$, $\ns{k}$, and $\sigma_i \leq j < k < \sigma_{i+1}$, $S,j \vDash \phi$ iff $S,k \vDash \phi$.
\end{lemma}

\begin{proof}
First of all, notice that, by constraint \ref{enum:c1}, $\sigma_i \geq \sigma_{i+1}-1$, $k > \sigma_i$ actually implies that $\ns{k}$; the only case in which it can be $\st{j}$ is when $j = \sigma_i$ and $\st{\sigma_i}$.

The proof proceeds by induction on the structure of $\phi$.

If $\phi = p \in AP$, then $p \in \beta(j)$ iff $p \in \beta(k)$, as $\beta(j) = \beta(k)$ by definition of $\sigma$, hence the result.

The cases $\phi = \neg \psi$ and $\phi = \phi_1 \land \phi_2$ are trivial.

If $\phi = \Dist{\psi, 1}$, then both $S,j \nvDash \phi$ and $S,k \nvDash \phi$, as $\Dist{\psi, 1}$ is by convention false in non-standard instants. 
Similarly when $\phi = \Dist{\psi, -1}$.

If $\phi = \Dist{\psi, \epsilon}$, then $S,j \vDash \phi$ iff $S,j+\epsilon \vDash \psi$ and $S,k \vDash \phi$ iff $S,k+\epsilon \vDash \psi$.
Since $\sigma_i < j+\epsilon < k+\epsilon < \sigma_{i+1}$, $\ns{j+\epsilon}$ and $\ns{k+\epsilon}$, then by inductive hypothesis $S,j+\epsilon \vDash \psi$ iff $S,k+\epsilon \vDash \psi$, hence the result.

If $\phi = \Until{}{\psi_1, \psi_2}$, we have that $S,k \vDash \phi$ iff there is a $t \geq k$ s.t. $S,t \vDash \psi_2$, and for all $k \leq t' < t$ it is $S,t' \vDash \psi_1$.
By inductive hypothesis, for all $t'$, $t''$ s.t. $\sigma_i \leq j \leq t'' < k \leq t' < \sigma_{i+1}$ where $\ns{j}$, we have that $S,t' \vDash \psi_1$ iff $S,t'' \vDash \psi_1$.
Hence, $S,t' \vDash \psi_1$ holds for all $k \leq t' < t$ iff also for all $j \leq t'' < t$ it is $S,t'' \vDash \psi_1$.
Then, $S,k \vDash \phi$ iff $S,j \vDash \phi$.
The case $\phi = \Since{}{\psi_1, \psi_2}$ is similar.

If $\phi = \X{st}{\psi}$, $S,j \vDash \phi$ iff $S,\sigma_{i+1} \vDash \psi$, as $\st{\sigma_{i+1}}$.
We have also $S,k \vDash \phi$ iff $S,\sigma_{i+1} \vDash \psi$, hence the result.

If $\phi = \X{ns}{\psi}$, both $S,j \nvDash \phi$ and $S,k \nvDash \phi$, as $\st{\sigma_{i+1}}$.

\proofend
\end{proof}